\documentclass[a4paper,12pt]{article}

\makeatletter

\@addtoreset{equation}{section}
\makeatother

\usepackage{amsfonts}
\usepackage{amssymb}
\usepackage{mathrsfs}
\usepackage{amsmath,amsthm}
\usepackage{amsmath}
\usepackage{comment}

\usepackage{setspace}

\usepackage{cite}
\usepackage{graphicx}
\usepackage{color}

\usepackage{here}

\usepackage{indentfirst}

\usepackage{ulem}
\usepackage{cancel}

\begin{document}
	
	\newtheorem{thm}{Theorem}
	\newtheorem{lemma}{Lemma}
	\newtheorem{prop}{Proposition}
	\newtheorem{cor}{Corollary}
	\theoremstyle{definition}
	\newtheorem{defn}{Definition}
	\newtheorem{remark}{Remark}
	\newtheorem{step}{Step}

	\newcommand{\Cov}{\mathop {\rm Cov}}
	\newcommand{\Var}{\mathop {\rm Var}}
	\newcommand{\E}{\mathop {\rm E}}
	\newcommand{\const }{\mathop {\rm const }}
	\everymath {\displaystyle}

	\newcommand{\ruby}[2]{
		\leavevmode
		\setbox0=\hbox{#1}
		\setbox1=\hbox{\tiny #2}
		\ifdim\wd0>\wd1 \dimen0=\wd0 \else \dimen0=\wd1 \fi
		\hbox{
			\kanjiskip=0pt plus 2fil
			\xkanjiskip=0pt plus 2fil
			\vbox{
				\hbox to \dimen0{
					\small \hfil#2\hfil}
				\nointerlineskip
				\hbox to \dimen0{\mathstrut\hfil#1\hfil}}}}

	\def\qedsymbol{$\blacksquare$}

	\renewcommand{\refname }{References}
	
	\title{
		A Note on Subadditivity of Value at Risks (VaRs): A New Connection to Comonotonicity
		\footnote{
			This is the preprint version of the paper published in {\it Journal of Applied Probability} (First View).}
	}
	
	\author{Yuri Imamura
		\thanks{Department of Business Economics, School of Management, 
			Tokyo University of Science, 
			1--11--2 Fujimi, Chiyoda-ku, Tokyo 102-0071, Japan, 
			E-mail: \texttt{imamuray@rs.tus.ac.jp}}
		\and Takashi Kato
		\thanks{Association of Mathematical Finance Laboratory (AMFiL), 
			2--2--15 Minamiaoyama, Minato-ku, Tokyo 107-0062, Japan, 
			E-mail: \texttt{takashi.kato@mathfi-lab.com}}
		}
	
	\date {}
	\maketitle 

	\begin{abstract}
		In this paper, we provide a new property of value at risk (VaR), which is a standard risk measure that is widely used in quantitative financial risk management. We show that the subadditivity of VaR for given loss random variables holds for any confidence level if and only if those are comonotonic. 
		This result also gives a new equivalent condition for the comonotonicity of random vectors. 
	\end{abstract}
	
\section{Introduction}\label{sec_intro}

Value at risk (VaR) is a standard risk measure used to capture tail risks in quantitative financial risk management. 
Mathematically, VaR is defined as the left continuous version of the generalized inverse function $F^{-1}_X$ of the distribution function $F_X(x) = P(X \leq x)$ of a random variable $X : \Omega \rightarrow \mathbb {R}$, which represents an uncertain amount of financial loss, that is, 
\begin{align*} 
	\mathrm {VaR}_\alpha (X) := \inf \{x\in \mathbb {R}\ ; \ F_X(x)\geq \alpha \} \equiv F^{-1}_X(\alpha ), 
\end{align*}
where $\alpha \in (0, 1)$ represents the confidence level of a risk manager and $(\Omega , \mathcal {F}, P)$ is a given probability space. 
VaR is widely used in financial practice because it is tractable and easy to understand. 
However, VaR lacks subadditivity, which is one of the important properties for a risk measure to appropriately  reflect the effect of risk diversification. 
Here, a risk measure $\rho $ is called subadditive if and only if 
\begin{align}\label {eq_subadditivity}
	\rho (X+Y) \leq \rho (X) + \rho (Y)
\end{align}
for any random variables $X, Y$ defined on $(\Omega , \mathcal {F}, P)$. 
Note that we can easily find an example where the subadditivity of VaR is violated. For instance, if we assume that $p, q\in (0, 1)$ 
satisfy $(1-p)(1-q) \leq \alpha < 1 - \max \{p, q\}$ for a given $\alpha \in (0, 1)$, we see that 
$$
\mathrm {VaR}_\alpha (X + Y) > \mathrm {VaR}_\alpha (X)  + \mathrm {VaR}_\alpha (Y) 
$$
for a random variable $X$ (resp., $Y$) whose distribution is given by the Bernoulli distribution with parameter $p$ (resp., $q$). 

In this paper, we consider the subadditivity of VaRs from another perspective. Specifically, we investigate a pair of random variables $(X, Y)$ satisfying the subadditivity (\ref {eq_subadditivity}) with $\rho = \mathrm {VaR}_\alpha $ for ``any confidence level $\alpha \in (0, 1)$.'' 
More generally, we say that VaR is subadditive for a random vector $(X_1, \ldots , X_n)$ if and only if the inequality 
\begin{align}\label {eq_VaR_subadditivity}
	\mathrm {VaR}_\alpha \left(\sum ^n_{i=1}X_i\right) \leq \sum ^n_{i=1}\mathrm {VaR}_\alpha (X_i), \ \ \alpha \in (0, 1)
\end{align}	holds. 
Then we show that such a random vector $(X_1, \ldots , X_n)$ is limited to the case where $(X_1, \ldots , X_n)$ is comonotonic. 
This result also gives a new equivalent condition for the comonotonicity of random vectors. 

\section{Main results}\label {sec_main}
To begin, we introduce the definition of the comonotonicity of random vectors. 

\begin{defn}Let $n\in \mathbb {N}$. 
	\begin{itemize}
		\item [(i)] \ A set $A\subset \mathbb {R}^n$ is called comonotonic if and only if for any $(x_1, \ldots , x_n), (y_1, \ldots , y_n)\in A$ with $x_i < y_i$ for some $i = 1, \ldots , n$, it follows that $x_j \leq y_j$ for each $j = 1, \ldots , n$. 
		\item [(ii)] \ A random vector $(X_1, \ldots , X_n)$ is called comonotonic if and only if the support of the distribution of $(X_1, \ldots , X_n)$ is a comonotonic set. 
	\end{itemize}
\end{defn}

The following theorem is our main result: 

\begin{thm}\label {th_main}
	Assume that $(\Omega , \mathcal {F}, P)$ is atomless and let $X_1, \ldots , X_n$ be integrable random variables. Then the following statements are equivalent: 
	\begin{itemize}
		\item [$\mathrm {(i)}$] \ $(\ref {eq_VaR_subadditivity})$ holds
		\item [$\mathrm {(ii)}$] \ $(X_1, \ldots , X_n)$ is comonotonic. 
	\end{itemize}
\end{thm}

\begin{remark}[Equivalent conditions for comonotonicity]
	The following properties are well known as equivalent conditions for the comonotonicity of a random vector $(X_1, \ldots , X_n)$ (see \cite {Dhaene-et-al1, Dhaene-et-al2, McNeil-Frey-Embrechts}). 
	\begin{itemize}
		\item [(E1)] The dependency structure of $(X_1, \ldots, X_n)$ is represented by the following copula: 
		$$
		M(u_1, \ldots , u_n) = \min \{u_1, \ldots , u_n\}, \ \ u_1, \ldots , u_n\in [0, 1]. 
		$$
		\item [(E2)] There exists a uniformly distributed random variable $U$ such that $(X_1, \ldots , X_n)$ and $(F^{-1}_{X_1}(U), \ldots , \allowbreak F^{-1}_{X_n}(U))$ have the same distribution. 
		\item [(E3)] There exist a random variable $Z$ and non-decreasing functions $f_1, \ldots , f_n$ such that $(X_1, \ldots , \allowbreak X_n)$ and $(f_1(Z), \ldots , f_n(Z))$ have the same distribution. 
	\end{itemize}
	If $X_1, \ldots , X_n$ are integrable, the comonotonicity of $(X_1, \ldots , X_n)$ is equivalent to the following condition \cite {Cheung1, Cheung2, Mao-Hu}: 
	\begin{itemize}
		\item [(E4)] Each random vector 
		$(\tilde {X}_1, \ldots , \tilde {X}_n)\in \mathcal {R}(F_{X_1}, \ldots , F_{X_n})$ satisfies 
		$$
		\sum ^n_{i=1}\tilde {X}_i \leq _\mathrm {cx} \sum ^n_{i=1}X_i, 
		$$
		where $\mathcal {R}(F_{X_1}, \ldots , F_{X_n})$ is the Fr\'echet class with marginal distributions $F_{X_1}, \ldots , F_{X_n}$; that is, 
		$$
		\mathcal {R}(F_{X_1}, \ldots , F_{X_n}) = \left\{(\tilde {X}_1, \ldots , \tilde {X}_n)\ ; \ F_{\tilde {X}_i} = F_{X_i}, \ i = 1, \ldots , n\right\}, 
		$$
		and we write $X\leq _\mathrm {cx} Y$ for random variables $X$ and $Y$ if $\mathrm {E}[X] = \mathrm {E}[Y]$ and
		$$
		\mathrm {E}\left[\max \left\{X - c, 0\right\}\right] \leq \mathrm {E}\left[\max \left\{Y - c, 0\right\}\right], \ \ c\in \mathbb {R}.
		$$
		
	\end{itemize}
	Furthermore, for the square-integrable case, the Pareto optimality for the convex order $\leq _\mathrm {cx}$ and the optimal coupling condition under the Wasserstein distance (when $n=2$) have been separately studied as equivalent conditions for the comonotonicity. See \cite {Cuestaalbertos-et-al, Denuit-et-al} for details. 
	
	Our result implies that $(\ref {eq_VaR_subadditivity})$ is also equivalent to the comonotonicity of $(X_1, \ldots , X_n)$. 
	Moreover, we see that (\ref {eq_VaR_subadditivity}) implies
	\begin{align}\label {eq_VaR_additivity}
		\mathrm {VaR}_\alpha \left(\sum ^n_{i=1}X_i\right) = \sum ^n_{i=1}\mathrm {VaR}_\alpha (X_i), \ \ \alpha \in (0, 1). 
	\end{align}
	due to Theorem \ref{th_main} and the comonotonic subadditivity of VaRs, hence (\ref {eq_VaR_subadditivity}) and (\ref {eq_VaR_additivity}) are equivalent. 
\end{remark}

\begin{remark}[Elliptic distributions]
	Elliptic distributions are known to be consistent with the subadditivity of VaRs. 
	That is, if $(X_1, \ldots , X_n)$ is elliptically distributed, the subadditivity of VaRs 
	\begin{align}\label {eq_VaR_subadditivity_alpha}
		\mathrm {VaR}_\alpha \left(\sum ^n_{i=1}X_i\right) \leq \sum ^n_{i=1}\mathrm {VaR}_\alpha (X_i)
	\end{align} 
	is 
	guaranteed; thus, we can use the VaR as if it is a coherent risk measure when $(X_1, \ldots , X_n)$ follows an elliptic distribution, including normal distributions and t-distributions. 
	However, we must keep in mind that (\ref {eq_VaR_subadditivity_alpha}) does not always hold for ``all'' $\alpha \in (0, 1)$. 
	
	Indeed, Theorem 1 in \cite {Embrechts-McNeil-Straumann} implies that if $(X_1, \ldots , X_n)$ is elliptically distributed and square-integrable, (\ref {eq_VaR_subadditivity_alpha}) holds for each $\alpha \in [1/2, 1)$. 
	Similarly to the proof of Theorem 1 in \cite {Embrechts-McNeil-Straumann}, if $\alpha \in (0, 1/2]$, we can verify the opposite inequality; that is, we see the ``superadditivity'' of the VaR: 
	$$
	\mathrm {VaR}_\alpha \left(\sum ^n_{i=1}X_i\right) \geq \sum ^n_{i=1}\mathrm {VaR}_\alpha (X_i), \ \ \alpha \in \left (0, 1/2\right]. 
	$$
	Thus, if we further assume (\ref {eq_VaR_subadditivity}), we obtain
	$$
	\mathrm {VaR}_\alpha \left(\sum ^n_{i=1}X_i\right) = \sum ^n_{i=1}\mathrm {VaR}_\alpha (X_i), \ \ \alpha \in (0, 1/2]. 
	$$
	This implies that 
	\begin{align}\label {eq_temp_elliptic}
		\sigma _i\sigma _j(1 - \rho _{ij}) = 0, \ \ i, j = 1, \ldots , n, 
	\end{align}
	where $\sigma _i$ is the standard deviation of $X_i$ and $\rho _{ij}$ is the correlation of $X_i$ and $X_j$. 
	In other words, for any $i, j = 1, \ldots , n$, it holds that ``if both $X_i$ and $X_j$ are non-deterministic, $\rho _{ij} = 1$.'' 
	Then we can verify that $(X_1, \ldots , X_n)$ is comonotonic. 
	Indeed, if given $i = 1, \ldots, n$ such that $\sigma _i > 0$ and letting $Z = (X_i - \mathrm {E}[X_i])/\sigma _i$ and $f_j(z) = \mathrm {E}[X_j] + \sigma _jz$ for each $j = 1, \ldots , n$, we see that $(X_1, \ldots , X_n)$ has the same distribution as $(f_1(Z), \ldots , f_n(Z))$. 
\end{remark}

\begin{remark}[Integrability condition]
	In Theorem \ref {th_main}, the integrability of $X_1\ldots , X_n$ is an essential assumption.
	Indeed, if $X_1, \ldots , X_n$ are not integrable, that is, if their distributions are sufficiently fat-tailed,there exists a counterexample where $(\ref {eq_VaR_subadditivity})$ holds while $(X_1, \ldots , X_n)$ is not comonotonic\footnote{The authors thank Prof. Ruodu Wang for pointing this out.}. 
	For instance, if $X_1, \ldots , X_n$ are ``independent'' and identically distributed random variables whose marginal distributions are the sign-inverted Pareto distribution 
	\begin{align*}
		F_{-X_1}(x) = \cdots = F_{-X_n}(x) = 1 - x^{-\gamma } \ (x\geq 1), \ \ 0\ (x < 1)
	\end{align*}
	with $\gamma \in (0, 1]$, we see that (\ref {eq_VaR_subadditivity}) is valid for $X_1, \ldots , X_n$. 
	More generally, Theorem 1 in \cite {Chen-Embrechts-Wang} implies that if $-X_1, \ldots , -X_n$ are weakly negatively associated and identically distributed super-Pareto, $X_1, \ldots , X_n$ satisfy (\ref {eq_VaR_subadditivity}) (also see the arguments below Proposition 2 in \cite {Chen-Embrechts-Wang}). 
\end{remark}

\begin{proof}[Proof of Theorem \ref {th_main}]
	It suffices to show that (i) implies (ii). Put $S = X_1 + \cdots + X_n$. 
	Because $(\Omega , \mathcal {F}, P)$ is atomless, Lemma A.32 in \cite {Foellmer-Schied} implies that there is a uniformly distributed random variable $U$ such that $S = F^{-1}_{S}(U)$ a.s. Now we define 
	$$
	\hat{X}_i = F^{-1}_{X_i}(U), \ \ i = 1, \ldots , n. 
	$$
	Then, Lemma A.23 in \cite {Foellmer-Schied} implies that the distribution of $\hat {X}_i$ is the same as that of $X_i$ for each $i = 1, \ldots , n$. 
	Thus, we get
	\begin{align}\label {eq_exp_Z}
		\mathrm {E}[Z] = 0, 
	\end{align}
	where $Z = \hat{S} - S$ and $\hat{S} = \hat {X}_1 + \cdots + \hat{X}_n$. 
	
	Meanwhile, from condition (i), we have that 
	$$
	F^{-1}_{S}(\alpha ) \leq \sum ^n_{i=1}F^{-1}_{X_i}(\alpha ), \ \ \alpha \in (0, 1). 
	$$
	Substituting $\alpha = U$, we get 
	$$
	S = F^{-1}_{S}(U) \leq \hat{S} \ \  \mbox {a.s.}, 
	$$
	so $Z\geq 0$ a.s. Combining this with (\ref{eq_exp_Z}), we see that $Z$ must be equal to zero a.s.; that is, 
	\begin{align}\label {eq_S_same}
		\hat{S} = S \ \ \mbox {a.s.}
	\end{align}
	
	Because $(\hat {X}_1, \ldots,  \hat {X}_1)$ is comonotonic, condition (E4) holds for $(\hat {X}_1, \ldots , \hat {X}_n)$ and thus 
	\begin{align*}
		\sum ^n_{i=1}\tilde {X}_i \leq _\mathrm {cx} \hat {S}, \ \ (\tilde {X}_1, \ldots , \tilde {X}_n)\in \mathcal {R}(F_{X_1}, \ldots , F_{X_n}). 
	\end{align*}
	Combining this with (\ref {eq_S_same}) and using (E4) again, we conclude that $(X_1, \ldots , X_n)$ is comonotonic. 
	\end{proof}

	\paragraph{\bf Acknowledgments:}
	The research was partially supported by JSPS KAKENHI Grant Number 20K03731. The support is gratefully acknowledged.

\end{document}